%% file: main.tex
\documentclass[11pt]{article}
\usepackage{fullpage}
\usepackage{amsmath}
\usepackage{amssymb}
\usepackage{amsthm}
\usepackage{verbatim}
\usepackage{graphicx}
\usepackage{epsfig}
\usepackage{algorithmic}
\usepackage{algorithm}

\newcommand{\cuts}{\mathcal{S}}

\newcommand{\Steiner}{\operatorname{St}}

\newcommand{\dem}{D}
\newcommand{\dems}{\mathcal{D}}
\newcommand{\packets}{\Pi}

\newcommand{\trees}{\mathcal{T}}
\newcommand{\load}{\ell}
\newcommand{\MST}{T}

\newcommand{\SteinerF}{F}

\newcommand{\out}{\delta^+}

\newtheorem{theorem}{Theorem}
\newtheorem{fact}{Fact}
\newtheorem{lemma}[theorem]{Lemma}

\begin{document}

\title{Network Design with Coverage Costs}

\author{
Siddharth Barman\thanks{California Institute of Technology. \tt{barman@caltech.edu}.} 
\and
Shuchi Chawla\thanks{University of
  Wisconsin -- Madison. \tt{shuchi@cs.wisc.edu}.}  
\and
Seeun Umboh\thanks{University of
  Wisconsin -- Madison. \tt{seeun@cs.wisc.edu}.}  
}



\date{}
\maketitle

\input{abstract.tex}



\section{Introduction}
\label{sec:intro}
\input{intro.tex}


\section{Problem Definition}
\label{sec:prelim}
\input{preliminaries.tex}
\section{A $2$-approximation for the laminar demands setting}
\label{sec:lam}
\input{lam.tex}


\section{A logarithmic approximation for the sunflower demands setting}
\label{sec:unif}
\input{uniform.tex}

\bibliographystyle{plain}
\bibliography{main,rand,algos,embedding,smartre}


\end{document}

%% file: abstract.tex
\begin{abstract}
We study network design with a cost structure motivated by redundancy in data traffic. We are given a graph, $g$ groups of terminals, and a universe of data packets. Each group of terminals desires a subset of the packets from its respective source. The cost of routing traffic on any edge in the network is proportional to the total size of the distinct packets that the edge carries. Our goal is to find a minimum cost routing. We focus on two settings. In the first, the collection of packet sets desired by source-sink pairs is laminar. For this setting, we present a primal-dual based $2$-approximation, improving upon a logarithmic approximation due to Barman and Chawla (2012)~\cite{BC12}. In the second setting, packet sets can have non-trivial intersection. We focus on the case where each packet is desired by either a single terminal group or by all of the groups, and the graph is unweighted. For this setting we present an $O(\log g)$-approximation.

Our approximation for the second setting is based on a novel spanner-type construction in unweighted graphs that, given a collection of $g$ vertex subsets, finds a subgraph of cost only a constant factor more than the minimum spanning tree of the graph, such that every subset in the collection has a Steiner tree in the subgraph of cost at most $O(\log g)$ that of its minimum Steiner tree in the original graph. We call such a subgraph a group spanner. 
\end{abstract}


%% file: intro.tex
Some of the classical applications of the theory of algorithms are in transportation and commodity networks: how should commodities be transported from where they are manufactured to where they are consumed? How should pipelines be laid to be most effective at balancing costs with requirements? Questions such as these have lead to some of the most basic problems and theorems in the area of approximation algorithms: network flow, traveling salesman, Steiner tree, flow-cut gaps, etc. Over time, solutions to these problems have come to be applied to a different class of networks, namely communication networks. At a basic level, the problems in communication networks are similar: how should data be routed from its sources to its destinations? How should networks be designed to be able to handle different kinds of workload and traffic patterns? However, the underlying commodity in these networks -- data -- is fundamentally different from physical commodities. Unlike the latter, data can be compressed, encoded, or replicated, at virtually no cost. Network algorithms that do not exploit these properties fail to utilize the entire capacity of the network.


The last few years have seen a rapid growth in ``content aware'' network optimization solutions, both within the academic literature (see, e.g., \cite{smartre,spring_protocol-independent_2000}, and references therein) as well as in the form of commercial technologies~\cite{blue_coat, riverbed}. One of the functionalities that these technologies provide is to remove duplicate traffic from the network. In particular, every router in the network equipped with such a technology keeps track of recently seen traffic. When duplicates are detected, a single copy of the duplicated data is sent forward along with a short message containing instructions for replication at the next router. This defines a cost function on every link in the network, where the cost of carrying data is proportional to the number (or total size) of {\em distinct} packets that the link carries; in other words, it is a {\em coverage function} over the set of traffic streams that use the link. We study network design problems within this context.


We consider the following framework. We are given a weighted network, and 
multiple {\em commodities}, each with a source and several possible destinations that we collectively call terminals. Each commodity is composed of a number of different data packets drawn from a universe of packets; we call these sets of packets {\em demands}. Importantly, there is redundancy in traffic---different commodities may overlap in the sets of packets they contain, and so can benefit from using common routes. 
Our goal is to find a minimum cost routing for the given traffic matrix, assuming that we can buy bandwidth at a fixed rate on every edge. Formally, our solution specifies for each commodity a routing tree spanning all of the terminals for this commodity. The cost of this solution on any particular edge is proportional to the total size of the distinct packets that the edge carries. This problem was introduced in \cite{BC12} where it was called redundancy aware network design.


Network design with coverage costs displays the same short-routes-versus-shared-routes tradeoff present in several classical network design problems with nonlinear costs, such as rent-or-buy network design~\cite{GKR02, GKPR03}, access network design~\cite{AZ02}, and buy-at-bulk network design~\cite{AA97, GKR03, MMP08, Talwar02}. However there are fundamental differences. The buy-at-bulk cost model is inspired by economies of scale in a physical commodity network---the volume of traffic that an edge carries is the sum of the volumes that the different commodities impose on it and the routing cost on the edge is a concave function of the total volume of traffic. On the other hand, in our setting, the volume of traffic itself is lowered due to the inherent nature of data traffic. In particular, this means that the savings achieved depend on the contents of the traffic and not just its quantity. We not only need to bundle traffic streams as much as we can, but we also need to decide the right sets of traffic streams to bundle. 
Consequently, the approximability of the problem also depends on the extent and manner in which different commodities share packets. When every source-sink pair in the network demands a distinct packet, that is, there is no data redundancy in the network, the problem reduces to finding the shortest route for each pair. When all of the demands are identical, the problem reduces to finding a single optimal Steiner forest over all of the terminal sets.


In this paper we focus on two special cases of the network design problem with coverage costs---the {\em laminar demands} setting, and the {\em sunflower demands} setting. In the laminar demands setting the packet sets corresponding to the commodities form a laminar family: the packet sets of any two commodities are either completely disjoint or one contains the other. There is a natural hierarchy over commodities in this setting and any commodity can use for free an edge that is being used for another commodity that ``dominates'' it. So we may favor long routes for a commodity if those routes share edges with a dominating commodity, in comparison to shorter ones that do not share edges. Less intuitively, it may be useful to pick similar routes for two commodities with disjoint packets sets if a portion of the shared route can be used for a commodity that dominates both. Consequently, commodities that are higher up in the hierarchy are in some sense more important than commodities that are lower in the hierarchy.

Non-laminar settings, where packet sets can have arbitrary intersection, also display sharing of paths among similar as well as dissimilar commodities. However, we cannot exploit any natural ordering over commodities in determining which paths to use. Our second setting captures the complexity introduced by non-trivial intersections. In the {\em sunflower demands} setting, every collection of demands has the same intersection. In other words, there is a common set of packets that belongs to every commodity, and every other packet belongs to exactly one commodity. A simple example of this setting is where each demand is of the form $\{0,i\}$; here $0$ denotes the common packet, and $i$ denotes the packet belonging only to commodity $i$. Once again our goal is to construct a routing tree for each commodity of minimum total cost. The cost of the collection of routing trees has two components. The first corresponds to the total size of the union of the routing trees: we pay for the cost of routing the common packets on this entire subgraph. The second corresponds to the costs of the individual trees, weighted by the sizes of their respective unique packets.

A standard approach in network optimization is to approximate a given network by a subgraph that is much cheaper or sparser than the entire graph, and yet faithfully captures some essential property of the graph. For example, spanners~\cite{pettie2007low} are low-cost subgraphs that approximately capture shortest path distances between every pair of points in the graph. Likewise, cut- and flow-sparsifiers~\cite{moitra2009approximation,leighton2010extensions} are sparse subgraphs that approximate cuts and flows in the graph respectively. Network design with coverage costs defines another such graph sparsification problem that may be of independent interest. In particular, for a given solution to the network design problem, consider partitioning the edges into sets that carry a particular packet. Each such set is a Steiner forest over the terminal sets that demand that packet. Our goal is to find a solution that minimizes a weighted sum of the sizes of these Steiner forests. One way of doing so may be to find a subgraph that induces Steiner forests over each respective set of terminals corresponding to a single packet, that are simultaneously approximately minimal for their corresponding instances. This approach is particularly relevant for the sunflower demands setting. In that setting, the Steiner forest corresponding to the common packets is the entire subgraph itself, whereas the forest corresponding to packets unique to a commodity is simply the routing tree constructed for that commodity. We therefore ask: is there a subgraph that $\alpha-$approximates the size of the minimum Steiner forest over the union of all terminal sets, and at the same time induces a Steiner tree over each individual terminal set that is within a factor of $\beta$ of the smallest such tree? We call such a subgraph an {\em $(\alpha,\beta)$ group spanner}. Group spanners generalize spanners: if for every pair of nodes in the graph our instance contains a terminal set comprising of the two nodes, then a group spanner for the instance simultaneously approximates the shortest path distances between every pair of nodes. The factor $\beta$ is called the stretch of the spanner.

The main technical component in our approach for the sunflower demands setting is a construction for group spanners in unweighted graphs where the union of all terminal sets spans the entire graph\footnote{We note that the first assumption by itself, i.e. the graph is unweighted, is without loss of generality: since our approximation is with respect to the sizes of the subgraphs, and not with respect to the number of edges, we can break up each long edge into edges of equal size by introducing new nodes. However, the additional assumption that every vertex belongs to some terminal set disallows this sort of transformation.}. Our construction achieves an $(O(1),O(\log g))$ approximation, where $g$ is the number of commodities. This implies an $O(\log g)$ approximation for the sunflower demands setting under those assumptions. We leave open the problem of extending our construction to arbitrary weighted graphs. 


For the laminar demands setting we obtain a $2$-approximation in general graphs. To form intuition for this setting consider an instance with $k$ different packets and $k+1$ commodities: for $i\le k$ the demand set of commodity $i$ contains only packet $i$, and demand set of commodity $k+1$ contains all of the $k$ packets.  Suppose also that every commodity has a single source and a single sink. Then, one approach to solving the problem is to first find a least cost path for commodity $k+1$, and then find least cost paths for the remaining commodities using the edges in the first path for free. This approach misses solutions where a slightly longer path for commodity $k+1$ is much more cost efficient for the remaining commodities than the shortest path for $k+1$. An alternative is to first find shortest paths for commodities $1$ through $k$, and then find the least cost path for commodity $k+1$ that can use edges in previously picked paths at a cheaper cost. This misses solutions where picking slightly longer paths for commodities $1$ through $k$ leads to a greater sharing of the edges. The first approach is indeed the approach analyzed in \cite{BC12} for the special case of the problem where there is a single source that belongs to all of the terminal sets. That paper shows that in any single source laminar demands setting routing commodities in order of decreasing sizes of demand sets achieves an $O(\log k)$ approximation where $k$ is the number of different packets in the universe.




We extend and improve the result of \cite{BC12} to obtain a $2$-approximation for the laminar demands setting with arbitrary terminal sets. Our approach is a hybrid of the two described above. At a high level, we first consider commodities in {\em increasing} order of the sizes of their demand sets. However, instead of committing to a single path for each commodity before considering the next, we keep around a collection of all possible near-optimal paths for the smaller demand sets before considering choices for the larger demand sets. Then in a second pass, we finalize a single path (tree) for each commodity, considering commodities in {\em decreasing} order of sizes of their demand sets. That is, we commit to paths for the larger demand sets before finalizing paths for the smaller demand sets. In order to maintain a collection of all near-optimal paths efficiently we use a primal-dual approach. The duals constructed for each commodity give a succinct description of all possible short paths connecting the source and the sink for that commodity. After having constructed all of the duals, we perform a reverse delete step that finalizes paths for commodities starting from the one with the largest demand and moving on to smaller demand sets.

\subsection{Connections to other network optimization problems}

The cost structure in the network design problem we consider is uniform in the sense that costs on different edges are related through constant factors. Obtaining a randomized $O(\log n)$ approximation for network design problems with a uniform cost structure is often easy: we can use the tree embeddings of Bartal~\cite{Bartal96} and Fakcharoenphol et al.~\cite{FRT03} to convert the graph into a distribution over trees such that distances between nodes are preserved to within logarithmic factors in expectation. Then the expected cost of the optimal routing over the (random) tree is related within logarithmic factors to the cost of the optimal routing over the graph. Moreover, the problem is easy to solve on trees, because there is a unique path between every pair of nodes. We achieve much better approximation factors. For the laminar demands setting, we obtain a $2$-approximation. For the sunflower demands setting, our approximation factor is $O(\log g)$; note that $g$ is always at most $n$, and in most applications should be much smaller.

As mentioned earlier, network design with coverage costs is closely related but incomparable to other models of network design with uniform costs that display economies of scale. This includes, e.g., the uniform buy-at-bulk~\cite{AA97, GKR03, MMP08, Talwar02}, rent-or-buy~\cite{GKR02, GKPR03}, and access network design~\cite{AZ02, GMM00} problems. For all of these problems constant factor approximations are known in the uniform costs setting for the special case where all of the commodities share a common source. In the multi-commodity setting, i.e., with distinct sources and sinks, the rent-or-buy network design problem admits a 2-approximation~\cite{GKR02, GKPR03}, but the buy-at-bulk network design problem is hard to approximate within poly-logarithmic factors~\cite{andrews2004hardness}.

Cost models specific to communication networks have been considered before in network design. Hayrapetyan et al.~\cite{hayrapetyan2005network} study a single-source network design problem in which the cost on an edge is a monotone submodular function of the commodities that use the edge. They obtain an $O(\log n)$ approximation via tree embeddings~\cite{Bartal96, FRT03}, where $n$ is the number of vertices in the graph. The cost structure that we consider is a special case of the one in~\cite{hayrapetyan2005network} (coverage functions are submodular). However, unlike~\cite{hayrapetyan2005network} we assume that terminals sets are arbitrary (in particular, they do not share a common source). Moreover, we obtain stronger approximation guarantees.


Shmoys et al.~\cite{SSL04} study a facility location problem with a cost structure very similar to that in our sunflower demands setting. In their model, the cost of opening a facility has two components: a fixed cost (similar to the cost of routing the common packets in our setting), and a service specific cost (similar to the cost of routing other packets in our setting). They present a constant factor approximation for facility location with this cost structure. Svitkina and Tardos~\cite{ST06} further extend this to a facility location problem with hierarchical costs, again presenting a constant factor approximation. Extending our results to more general non-laminar coverage functions including hierarchical costs is an interesting open problem.




As mentioned earlier, a main component in our approach for the sunflower demands setting is a construction for group spanners in unweighted graphs. Group spanners generalize graph spanners. Low-stretch spanners have a number of applications, including distributed routing using small routing tables and in computing near-shortest paths in distributed networks (see \cite{pettie2007low} and references therein). In unweighted graphs it is well known that the size of the smallest spanner with multiplicative stretch $k$ is equal to the maximum number of edges in a graph with girth at least $k+1$; this is known to be $O(n^{1+O(1/k)})$, and is conjectured tight. Our result is consistent with this bound: when the number of commodities $g$ is equal to the number of vertex pairs, we get an $O(\log g)=O(\log n)$ stretch with a spanner of size $O(n)$. Other work on spanners has focused on additive stretch and weighted graphs (see, e.g., \cite{elkin2004, pettie2007low, thorup2006spanners}). 

Group spanners also generalize shallow-light spanning trees. The latter is a subgraph that is simultaneously an approximately-minimum spanning tree of the given graph, as well as an approximate-shortest-paths tree with respect to a given source node. Consider an instance with a special source node $s$ that for every node $v$ in the graph contains the terminal set $\{s,v\}$. Then an $(\alpha,\beta)$ group spanner for this instance simultaneously approximates the shortest path distance from $s$ to $v$ for every $v$ to within a factor of $\beta$, and has size no more than $\alpha$ times the size of the minimum spanning tree in the graph. However, while our approach only guarantees $\beta=O(\log n)$ for $g=n$ commodities, it is possible to obtain an $(O(1/\epsilon),1+\epsilon)$ approximation for any $\epsilon>0$ \cite{ABP90,KRY93}. 



%% file: preliminaries.tex
In this section, we formally define Network Design with Coverage Costs. 
We are given a graph $G = (V,E)$ with costs $c_e$ on edges, a universe $\packets$ of packets, and $g$ commodities with terminal sets $X_1, \ldots, X_g \subseteq V$. The demand set of terminal set $X_j$ is denoted $\dem_j\subseteq \packets$, and we denote the collection of all demand sets as $\dems$. A solution consists of a collection of $g$ Steiner trees $\trees= \{T_1, \ldots, T_g\}$ where $T_j$ is a Steiner tree spanning terminal set $X_j$. The trees specify how packets are to be routed over the edges: the packets of demand $D_j$ are routed over edges of $T_j$.  For a solution $\trees$, the load on edge $e$ is $\load_e(\trees) = |\bigcup_{i : e \in T_i} D_i|$, i.e. the total number of distinct packets being routed over edge $e$. More generally, we can consider a setting in which packets have weights and we define the load on an edge to be the total weight of all of the distinct packets that an edge carries. The performance and running times of both of our algorithms are independent of the number of distinct packets, so we may assume without loss of generality that all packets have unit weight. Our goal is to find a solution $\trees$ so as to minimize the total cost $\sum_{e \in E} c_e \load_e(\trees)$.

We now describe the two special cases of network design with coverage costs that we study. In the following, for a subgraph $H$, we write $c(H)$ for the total cost of edges in $H$, i.e. $c(H) := \sum_{e \in H} c_e$.  

\paragraph{Laminar demands.} In this setting, the collection of demand sets is laminar: for any $\dem,\dem' \in \dems$, $\dem \cap \dem' \neq \emptyset$ implies either $\dem \subseteq \dem'$ or $\dem' \subseteq \dem$. In this case we can transform our objective into a simpler form where the cost of each edge is charged to a collection of {\em disjoint} demand sets. In particular, given a solution $\trees$, for an edge $e$ consider the demand sets $D$ that are maximal among the collection $\{\dem_j : e \in T_j\}$ of demand sets that this edge carries. Because of laminarity, these maximal demand sets are disjoint, and so the load on the edge is simply the sum of the sizes of these demand sets. Accordingly, let us define $H_\dem(\trees)$ to be the set of edges $e$ such that $\dem$ is a maximal set in $\{\dem_j : e \in T_j\}$. The packet set $\dem$ will contribute to the load on these edges. Then we can write the total cost of the solution $\trees$ as \[\load(\trees) = \sum_e c_e \load_e(\trees) = \sum_e \sum_{\dem:
  H_\dem(\trees) \ni e} c_e|\dem| = \sum_D |D| \sum_{e\in
  H_\dem(\trees)} c_e = \sum_D |D| c(H_\dem(\trees) ).\]

Further note that in a feasible solution $\trees$, for each commodity
$j$, the subgraph $\bigcup_{\dem \supseteq \dem_j} H_{\dem}(\trees)$ contains the tree $T_j$ and therefore
spans the terminal set $X_j$. Therefore, instead of specifying a Steiner tree
for each terminal set, it suffices to specify a forest $H_\dem$ for
each demand set $\dem$ such that each terminal set $X_j$ is connected in $\bigcup_{\dem \supseteq \dem_j} H_{\dem}$.

\paragraph{Sunflower demands.}
In this setting, there is a special set of packets $P\subseteq \Pi$ such that for all $i\ne j$, we have $D_i \cap D_j = P$. In other words,  $D_j = P \cup P_j$ with $P_i \cap P_j = \emptyset$ for all $i\ne j$. We can again transform our objective into a simpler form. For a routing solution $\trees = \{ T_1, T_2, \ldots, T_g\}$, let $H$ denote the subgraph obtained by taking the union of the $T_j$s. 
Observe that $H$ is a Steiner forest for $X_1, \ldots, X_g$. We have to route $P$ over $H$, since all terminal sets demand $P$, and $P_j$ over $T_j$. Thus the cost of the routing solution can be expressed as $\load(\trees) = |P| c(H) + \sum_j |P_j| c(T_j)$.
%
%

We will now describe a lower bound on the cost of the optimal solution in this setting. For a vertex set $X$ and subgraph $H$, let $\Steiner_H(X)$ denote the cost of an optimal (i.e., minimum cost) Steiner tree over $X$ in $H$. Let $\trees^* = \{T_1^*, T_2^*, \ldots, T_g^*\}$ be an optimal routing solution to the given instance and let $H^*=\bigcup_j T_j^*$. Suppose $\SteinerF^*$ is an optimal Steiner forest for $X_1, \ldots, X_g$. Since $H^*$ is a Steiner forest for $X_1, \ldots, X_g$ and $T_j^*$ is a Steiner tree for $X_j$, we have $c(H^*) \geq c(\SteinerF^*)$ and  $c(T_j^*) \geq \Steiner_G(X_j)$. Therefore the optimal routing-solution cost can be bounded as $\load(\trees^*) \geq |P| \  c\left( \SteinerF^* \right) +  \sum_j |P_j| \ \Steiner_G(X_j)$.



\paragraph{Group spanners.}
For a graph $G = (V,E)$ with cost $c_e$ on edges and $g$ terminal sets $X_1, \ldots, X_g \subseteq V$, we say that subgraph $H$ is an \emph{$(\alpha, \beta)$ group spanner} if $c(H) \leq \alpha c(F^*) $ and $\Steiner_H(X_j) \leq \beta \Steiner_G(X_j)$ for all $j$. Here $F^*$ denotes an optimal Steiner forest for $X_1, \ldots, X_g$ in $G$.  Note that a group spanner generalizes the notion of a spanner since the latter asks for a sparse spanning subgraph $H$ such that for every pair of vertices $(u,v)$ we have $\beta$ stretch: $d_H(u,v) \leq \beta d_G(u,v)$. Here $d_H(u,v)$ (respectively, $d_G(u,v)$) denotes the distance, with edge lengths $c_e$, between vertices $u$ and $v$ in $H$ (respectively, $G$).

The following lemma shows that a good group spanner implies an approximation for the sunflower demands setting. 


\begin{lemma}\label{lem:sparse}
Given an $(\alpha, \beta)$ group spanner $H$ for graph $G$ and terminal sets $X_1, X_2, \ldots, X_g$, we can obtain an $\alpha + 2\beta$ approximation for any sunflower demands instance defined over $G$ and $X_j$s.   
\end{lemma}
\begin{proof}
 For all $j$, let $H_j$ be the Steiner trees over $X_j$ in $H$ obtained via any constant factor approximation. We set $\{H_1, H_2, \ldots, H_g \}$ as the routing solution for the given instance. The cost of this solution is no more than $|P| c(H) + \sum_j |P_j| c(H_j)$. Recall that the optimal routing-solution cost for the given instance is at least $ |P| \  c\left( \SteinerF^* \right) +  \sum_j |P_j| \ \Steiner_G(X_j) $. Therefore, using the fact that $H$ is an $(\alpha, \beta)$ group spanner and $c(H_j) \leq O(1) \Steiner_H(X_j)$, we get the desired claim.

 \end{proof}

Note that using group spanners we get an oblivious approximation in the sense that the construction uses only the knowledge of the underlying graph and the terminal sets but not the demand sets.

In Section~\ref{sec:unif} we consider unweighted graphs with terminal sets that satisfy $V= \bigcup_j X_j$. We develop an algorithm that obtains a $(14, O(\log g) )$ group spanner for such an instance, and so by Lemma~\ref{lem:sparse} gives an $O(\log g)$ approximation to the sunflower demands setting over the instance (see Theorem~\ref{thm:uniform-main}).  



%% file: lam.tex
Recall that in the laminar demands setting, for all $\dem,\dem' \in
\dems$ with $\dem \cap \dem' \neq \emptyset$, we have $\dem \subseteq
\dem'$ or $\dem' \subseteq \dem$. As established in
Section~\ref{sec:prelim}, in order to obtain a feasible solution in
this setting, it suffices to specify a forest $H_\dem$ for each demand
set $\dem$ such that each terminal set $X_j$ is connected in
$\bigcup_{\dem \supseteq \dem_j} H_{\dem}$. The cost of the
corresponding routing is $\sum_D |D| c(H_\dem(\trees) )$.

Our algorithm for the laminar demands case is an extension of the
Goemans-Williamson primal-dual algorithm for the Steiner Forest
Problem \cite{GW95}. We begin by defining the primal and dual linear
programs.

In the linear program below, the variable $x_{e,\dem}$ denotes whether
$e \in H_\dem$. We denote by $\delta(S)$ the set of edges crossing a
cut $S \subseteq V$, and by $\cuts_\dem$ the collection of cuts $S
\subseteq V$ that separates a terminal set $X_j$ with
$D_j \supseteq\dem$. The cut constraints require that each terminal
set $X_j$ is connected by $\bigcup_{\dem \supseteq \dem_j} H_{\dem}$.


\begin{equation*}
\boxed{
\begin{aligned}
  \mbox{minimize}\quad            
  & \sum_{e,\dem \in \dems} x_{e,\dem}\cdot |\dem|c_e\\
  \mbox{subject to}\quad  
  & \sum_{\dem' \supseteq \dem} \sum_{e \in \delta(S)} x_{e,\dem'} \geq 1 &\quad
  \forall \dem \in \dems, S \in \cuts_\dem
\end{aligned} 
}
\end{equation*}

The corresponding dual linear program is as follows.
\begin{equation*}
\boxed{
\begin{aligned}
  \mbox{maximize}\quad            
  & \sum_{\dem \in \dems, S \in \cuts_\dem} y_{\dem,S} \\
  \mbox{subject to}\quad  
  & \sum_{\dem' \subseteq \dem} \sum_{S \in \cuts_{\dem'} : e \in \delta(S)} y_{\dem',S} \leq |\dem|c_e &\quad
  \forall e, \dem \in \dems
\end{aligned} 
}
\end{equation*}

\subsection{Algorithm}
The algorithm starts with a dual ascent stage in which it adds edges
to forests $\{F_\dem\}_{\dem \in \dems}$, and ends with a pruning
stage. In the following discussion, for a demand set $\dem \in \dems$
we say that $S \in \cuts_\dem$ is a \emph{$\dem$-unsatisfied cut} if
$(\bigcup_{\dem' \supseteq \dem} F_{\dem'}) \cap \delta(S) = \emptyset$.
We also say that an edge $e$ is \emph{$\dem$-tight} if
\[\sum_{\dem' \subseteq \dem} \sum_{S \in \cuts_{\dem'} : e \in
  \delta(S)} y_{\dem',S} = |\dem|c_e.\]

In the dual ascent stage, the algorithm raises duals in phases, one
per demand set $\dem \in \dems$ in order of increasing size.
In phase $\dem$, while there exists a $\dem$-unsatisfied cut
it alternates between raising duals of the minimal $\dem$-unsatisfied
cuts and adding $\dem$-tight edges to $F_\dem$. We say that $S$ is an
\emph{active set} in the current iteration of the inner while loop if
it is a minimal $\dem$-unsatisfied cut. The algorithm ensures that at
the end of phase $\dem$, the edges $F_\dem$ are paid for by the dual
and $F_\dem$ is a Steiner forest for terminal sets whose demand set
contains $\dem$.  In the pruning stage, the algorithm processes the
demand sets in order of decreasing size and removes unnecessary edges
from $\{F_\dem\}_{\dem \in \dems}$ and returns $\{H_\dem\}_{\dem \in
  \dems}$.

\begin{algorithm}
\caption{Primal-Dual Algorithm for Laminar Buy-at-Bulk}
  \begin{algorithmic}[1]
   \label{alg:laminar}
    \STATE Initialize $F_\dem \leftarrow \emptyset$ for all $\dem \in \dems$ and
    $y_{\dem,S} \leftarrow 0$ for all $\dem \in \dems, S \subseteq V$.
    \STATE \emph{(Dual ascent stage)}
    \FOR {$\dem \in \dems$ in increasing order of size}
    \STATE \emph{(Start of phase $\dem$)}
    \WHILE {there exists a $\dem$-unsatisfied cut}
    \STATE Simultaneously raise $y_{\dem,S}$ for active sets $S$
    until some edge $e$ goes $\dem$-tight.
    \STATE $F_\dem \leftarrow F_\dem + e$.
    \ENDWHILE
    \STATE \emph{(End of phase $\dem$)}
    \ENDFOR
    \STATE \emph{(End of dual ascent stage)}
    \STATE \emph{(Pruning stage)}
    \STATE $H_\dem \leftarrow F_\dem$ for all $\dem \in \dems$.
    \FOR {$\dem \in \dems$ in decreasing order of size}
    \FOR {$e \in H_\dem$}
    \IF {$(H_\dem - e ) \cup \bigcup_{\dem' \supsetneq \dem} H_{\dem'}$ is a Steiner
      forest for terminal sets with demand set $\dem$}
    \STATE $H_\dem \gets H_\dem - e$.
    \ENDIF
    \ENDFOR
    \ENDFOR
    \STATE \emph{(End of pruning stage)}
    \RETURN $\{H_\dem\}_\dem$
  \end{algorithmic}
\end{algorithm}

The following lemma implies that we can efficiently find active
sets. 
\begin{lemma}
  \label{lem:active}
  In any iteration in phase $\dem$, a set $S$ is active if and only if
  it is a component of $F_\dem$ and it separates a terminal set whose
  demand set contains $\dem$.
\end{lemma}

\begin{proof}
   Let $S$ be an active set. By definition, $S$ is a minimal cut in
   $\cuts_\dem$ such that $\bigcup_{\dem' \supseteq \dem} F_{\dem'}
   \cap \delta(S) = \emptyset$. Since $S \in \cuts_\dem$, it separates
   a terminal set whose demand set contains $\dem$. The algorithm raises duals for demand sets in increasing order of size, so we have
 $F_{\dem'} = \emptyset$ for $\dem' \supsetneq \dem$. This implies
 that $F_{\dem} \cap \delta(S) = \emptyset$ and so $S \cap C =
 \emptyset$ or $S \cap C \supseteq C$ for every connected component
 $C$ of $F_{\dem}$. Thus, $S$ is a superset of a union of connected components of $F_{\dem}$.
The algorithm processes the demand sets in increasing order of size,
   so we have $F_{\dem'} = \emptyset$ for $\dem' \supsetneq \dem$ and
   thus $F_{\dem} \cap \delta(S) = \emptyset$. This implies that $S
   \cap C = \emptyset$ or $S \cap C \supseteq C$ for every connected
   component $C$ of $F_{\dem}$ and so $S$ is a superset of a union of
   connected components of $F_{\dem}$.  By minimality, we have that $S$
   is a connected component of $F_{\dem}$.

   For the converse, consider a connected component $S'$ of $F_{\dem}$
   that separates a terminal set whose demand set contains $\dem$. By
   definition, we have $S' \in \cuts_\dem$. Since $S'$ is a connected
   component of $F_\dem$ and $F_{\dem'} = \emptyset$ for $\dem'
   \supsetneq \dem$, it is a minimal set in $\cuts_\dem$ such that
   $\bigcup_{\dem' \supseteq \dem} F_{\dem'} \cap \delta(S) =
   \emptyset$. Therefore $S'$ is an active set.
 \end{proof}

\subsection{Analysis}
Our analysis follows along the lines of the analysis for the
Goemans-Williamson algorithm. We first establish that the primal and dual
solutions generated by the algorithm are feasible. 
\begin{lemma}
  \label{lem:feasible}
  The primal solution $\{H_\dem\}_{\dem \in \dems}$ and the dual
  solution $\{y_{\dem,S}\}_{\dem \in \dems, S \subseteq V}$
  are feasible.
\end{lemma}

 \begin{proof}
   We first prove that the primal solution is feasible. Consider an
  iteration during the pruning stage. We say that terminal set $X_j$
  is \emph{$H$-disconnected} if it is disconnected with respect to
  edge set $\bigcup_{\dem \supseteq \dem_j} H_\dem$ and
  \emph{$H$-connected} otherwise. We will show that all terminal sets
  are $H$-connected in all iterations of the pruning stage.

  Observe that at the end of phase $\dem$, there are no
  $\dem$-unsatisfied cuts and $F_{\dem'} = \emptyset$ for $\dem'
  \supsetneq \dem$. Thus, all terminal sets with demand set $\dem$ are
  connected with respect to edge set $F_{\dem}$. At the beginning of
  the pruning stage, we have 
  $H_\dem = F_\dem$ for all $\dem \in \dems$, and so all
  terminal sets are $H$-connected. Consider an iteration in
  which the algorithm deletes an edge $e$ from $H_\dem$.
  %
  %
  By definition of $H$-disconnected, this can only cause a terminal
  set with demand set $\dem' \subseteq \dem$ to be
  $H$-disconnected. However, the algorithm will not delete $e$ if it
  causes a terminal set with demand set $\dem$ to be
  $H$-disconnected. Now consider a demand set $\dem' \subsetneq
  \dem$. Since $|\dem'| \leq |\dem|$, we still have $H_{\dem'} =
  F_{\dem'}$ so all terminal sets with demand set $\dem'$ are
  $H$-connected. Thus, all terminal sets are $H$-connected throughout
  the pruning stage and so $\{H_\dem\}_{\dem \in \dems}$ is a feasible
  primal solution.

  The dual solution is feasible since the algorithm explicitly ensures
  that the dual variables in a tight constraint are not raised.
 \end{proof}

Next, we show that in each phase $\dem$ of the dual raising stage,
the current active sets has average degree with respect to edges
$\bigcup_{\dem' \supseteq \dem} H_{\dem'}$ (formally defined below) at most $2$ in every
iteration. This in turn implies that the primal solution has cost at
most twice the total dual value. Since the dual is feasible, we have
that the algorithm gives a $2$-approximation. We bound the average
degree of active sets by showing that $\bigcup_{\dem' \supseteq \dem}
H_{\dem'}$ is a forest and that no inactive set has degree $1$. 
\begin{lemma}
  \label{lem:forest}
  For all $\dem \in \dems$, we have that $\bigcup_{\dem' \supseteq \dem}
  H_{\dem'}$ is a forest.
\end{lemma} 

\begin{proof}
Suppose, towards a contradiction, that the statement is false. Let
  $\dem$ be a maximal demand set such that $\bigcup_{\dem' \supseteq
    \dem} H_{\dem'}$ contains a cycle $C$.  By maximality, there
  exists $e \in C \cap H_\dem$. Since $e$ is in a cycle in
  $\bigcup_{\dem' \supseteq \dem} H_{\dem}$, we have that $(H_\dem - e) \cup \bigcup_{\dem' \supsetneq \dem} H_{\dem'}$ is still a Steiner
  forest for terminal sets with demand set $\dem$. Thus, the algorithm
  would have removed $e$ from $H_{\dem}$ and so we have a
  contradiction.
  
 \end{proof}


For a subset of edges $E' \subseteq E$, let $\deg_{E'}(S) = |\delta(S)
\cap E'|$ denote the number of edges in $E'$ exiting $S$.
\begin{lemma}
  \label{lem:inactive}
  Consider an iteration in phase $\dem$ of the dual raising stage. Let
  $S$ be a connected component of $F_\dem$ in this iteration. If $S
  \notin \cuts_\dem$, then $\sum_{\dem' \supseteq \dem}
  \deg_{H_{\dem'}}(S) \neq 1$.
\end{lemma}

\begin{proof}
We prove the contrapositive. Suppose $\sum_{\dem' \supseteq \dem}
  \deg_{H_{\dem'}}(S) = 1$. Let $e$ and $A \supseteq \dem$ be the
  unique edge and demand set, respectively, such that $e \in H_{A}
  \cap \delta(S)$. Since the algorithm did not delete $e$ from $H_{A}$
  and $\bigcup_{\dem' \supseteq A}H_{\dem'}$ is acyclic by Lemma \ref{lem:forest},
  there exists $X_j$ with $D_j = A$ and $u,v \in X_j$ such that $e$ is
  on the unique $u-v$ path in $\bigcup_{\dem' \supseteq
    A}H_{\dem'}$. Since $\sum_{\dem' \supseteq \dem}
  \deg_{H_{\dem'}}(S) = 1$, the path crosses $S$ exactly once. Thus,
  we have that $S$ separates $u,v$ and so $S \in \cuts_A$. 
  By definition of $\cuts_\dem$, we have $\cuts_A \subseteq
  \cuts_\dem$ and this completes the proof of the lemma.

 \end{proof}

We are now ready to prove that the primal solution has cost at most
twice the dual value.
\begin{lemma}
  \label{lem:dual}
  $\sum_\dem \sum_{e \in H_\dem} |\dem|c_e \leq 2 \sum_{\dem,S} y_{\dem,S}$.
\end{lemma}

\begin{proof}
Using the fact that we only add tight edges, we have
\begin{align*}
  \sum_\dem \sum_{e \in H_\dem} |\dem|c_e 
  &= \sum_\dem \sum_{e \in H_\dem} \left(\sum_{\dem' \subseteq \dem}
    \sum_{S \in \cuts_{\dem'}:
      e \in \delta(S)} y_{\dem',S}\right) \\
 &  = \sum_{\dem'} \sum_{S \in \cuts_{\dem'}} y_{\dem',S} \left(\sum_{\dem \supseteq \dem'}
    \sum_{e \in
      \delta(S) \cap H_\dem} 1\right) \\
  & = \sum_{\dem'} \sum_{S \in \cuts_{\dem'}} y_{\dem',S} \left(\sum_{\dem \supseteq \dem'}
    \deg_{H_\dem}(S)\right) \\ 
 &  = \sum_{\dem'} \sum_{S \in \cuts_{\dem'}} y_{\dem',S} \deg_{\bigcup_{\dem \supseteq \dem'}
    H_\dem}(S).
\end{align*}
The second equality is obtained by rearranging, and the last follows from
the fact that each edge is in $H_{\dem}$ for at most one $\dem
\supseteq \dem'$.

Suppose that in an iteration in phase $\dem'$, the dual for each
active set is raised by $\Delta$. This implies $\sum_{S \in
  \cuts_{\dem'}} y_{\dem',S} \deg_{\bigcup_{\dem \supseteq \dem'}
  H_\dem}(S)$ increases by $\Delta \cdot \sum_{S \text{ active}}
\deg_{\bigcup_{\dem \supseteq \dem'} H_\dem}(S)$, and $\sum_{\dem,S}
y_{\dem,S}$ increases by $\Delta \cdot \text{\# active sets}$. So
it suffices to prove that in each phase $\dem'$ and in each iteration
within the phase, the average degree of active sets is at most $2$:
\[\sum_{S \text{ active}} \deg_{\bigcup_{\dem \supseteq \dem'} H_\dem}(S) \leq 2 \cdot
\text{\# active sets}.\]

Fix an iteration in phase $\dem'$. 
Note that each active set corresponds to some connected component of
$F_{\dem'}$ by Lemma \ref{lem:active}.
%
Let $G'$ be a graph whose nodes are connected components of
$F_{\dem'}$ and whose edge set is $\bigcup_{\dem \supseteq \dem'}
H_\dem$. The degree of a node in $G'$ is equal to the degree of the
corresponding set with respect to edge set $\bigcup_{\dem \supseteq
  \dem'} H_\dem$. Let us say that a node of $G'$ corresponding to an
active set is an \emph{active node}, and that any other node is
\emph{inactive}. We want to show that the average degree of active
nodes in $G'$ is at most $2$. Suppose we remove all isolated nodes
from $G'$. In the resulting graph, by Lemma \ref{lem:inactive} the
degree of each inactive node is at least $2$, and by Lemma
\ref{lem:forest} the average degree is at most $2$. So the claim follows.
\end{proof}

Lemmas \ref{lem:feasible} and \ref{lem:dual} gives us the following
theorem.
\begin{theorem}
  Algorithm \ref{alg:laminar} is a $2$-approximation for network
  design with coverage costs in the laminar demands setting.
\end{theorem}


%% file: uniform.tex
We now consider the sunflower demands setting. The main technical result of this section is the following lemma which
says that we can find a group spanner 
of linear size with stretch
$O(\log g)$.
\begin{lemma}
  \label{thm:uniform-unweighted}  
  Given an unweighted graph $G = (V, E)$ ($c_e = 1$ for all $e \in G$)
  and terminal sets $X_1, \ldots, X_g$ such that $V= \bigcup_j X_j$,
  we can construct in polynomial time a $(14, 4\log g)$ group spanner.
\end{lemma}

Before we prove Lemma~\ref{thm:uniform-unweighted}, we observe that,
together with Lemma~\ref{lem:sparse}, it implies the following result
for unweighted instances of the sunflower demands setting with vertex
set $V = \bigcup_j X_j$.

\begin{theorem}\label{thm:uniform-main}
  Network design with coverage costs in the sunflower demands setting
  admits an $O(\log g)$ approximation over unweighted graphs with
  vertex set $V = \bigcup_j X_j$.
\end{theorem}

In the remainder of the section we will focus on unweighted graphs and
write $|H|$ to denote the cost (i.e., the number of edges) of subgraph
$H$. Let us recall some notation: for a subgraph $H$, $\Steiner_H(X)$
denotes the cost of an optimal (i.e., minimum cost) Steiner tree over
vertex set $X$ in $H$, and $d_H(u,v)$ denotes the distance between
vertices $u,v$ in $H$.  Let $\MST$ denote a minimum spanning tree of
the given graph $G$.


Now we prove Lemma~\ref{thm:uniform-unweighted}.  To that end we
consider \emph{uniform} group spanner instances where the following
holds for all $j$: for all strict subsets $S$ of
$X_j$, 
there exists an edge $(x,y) \in E$ such that $x \in S, y \in X_j
\setminus S$. In other words, there exists an optimal Steiner tree for
each $X_j$ with no Steiner vertices and it is easy to find. 

Next we show that in order to establish
Lemma~\ref{thm:uniform-unweighted} it suffices to solve uniform
instances.  We can transform any given group spanner instance
over an unweighted graph $G$ with $V = \bigcup_j X_j$ into a uniform
instance as follows: add to $X_j$ all Steiner vertices in the
$2$-approximate Steiner tree given by the MST
heuristic~\cite{vazirani2001approximation} applied over $X_j$ in
$G$ 
and let $X'_j$ be the resulting set. Since $X'_j$ is the set of all vertices
of a Steiner tree, the group spanner instance with terminal sets
$X'_1, \ldots, X'_g$ is a uniform one.

Say we obtain subgraph $H$ after solving the above uniform instance
and $H$ satisfies $\Steiner_H(X'_j) \leq \beta \Steiner_G(X'_j)$ for
all $j$ and $|H| \leq \alpha |\MST|$.  We show that $H$ is in fact a
$(2\alpha, 2 \beta)$ group spanner for the original instance.  The
MST heuristic guarantees that $\Steiner_G(X_j') \leq 2
\Steiner_G(X_j)$; which implies $\Steiner_H(X_j) \leq 2 \beta
\Steiner_G(X_j)$.  Finally, let $\SteinerF^*$ denote an optimal Steiner
forest for $X_1, \ldots, X_g$ in $G$. In an unweighted instance, we have that
$|\SteinerF^*| \geq |\MST|/2$. This is because $V= \bigcup_j X_j$ and
each component of the forest has at least one edge\footnote{We assume
  without loss of generality that $|X_j| \geq 2$ for all $j$} so
$|\SteinerF^*| \geq |V|/2 \geq |\MST|/2$. Since, $|H| \leq \alpha
|\MST|$ we get the cost guarantee, $|H| \leq 2 \alpha |\SteinerF^*|$.


This implies that to prove Lemma~\ref{thm:uniform-unweighted} we only need to solve uniform group spanner instances.  In the remainder of this section, we focus on uniform instances and for ease of exposition write $X_j$ in place of $X_j'$.

\begin{lemma}
  \label{thm:uniform}
  Given any uniform group spanner instance with terminal sets $X_j$, there exists a subset of edges $A$ of size
  $|A| \leq 6 |\MST|$ such that for $H:=A \cup \MST $ we have $\Steiner_{H}(X_j) \leq (2\log
  g)  \Steiner_{G}(X_j)$ for all $j$.
\end{lemma}


Since $|H| = |A| + |T|  \leq 7 |\MST|$ and $\Steiner_{H}(X_j) \leq
(2\log g) \Steiner_{G}(X_j)$, we get that $H$ is a $(14, 4 \log g)$ group spanner that satisfies the desired bounds in Lemma~\ref{thm:uniform-unweighted}.

We now move on to present a constructive proof of
Lemma~\ref{thm:uniform}. We assume that terminals of $X_j$ are ordered
$x_{j,1}, x_{j,2}, \ldots$ such that for $i > 1$, there exists an edge
$(x_{j,i}, x_{j,k}) \in E$ for some $k < i$; we call this edge a
\emph{satisfying edge} for $x_{j,i}$. For ease of notation, we drop
the indices when they do not matter and write $(x,y)$ to denote $x$'s
satisfying edge. Note that such an ordering always exists, e.g. a
preordering of the (uniform) Steiner tree over $X_j$ with any root. We say that a terminal
$x_{j,i} \in X_j$ is \emph{unsatisfied}\footnote{We define the lowest
  indexed vertex $x_{j,1}$ to be always satisfied.} in a spanning
subgraph $H$ if $d_H(x_{j,i}, \{x_{j,1} \ldots, x_{j,i-1}\}) > 2\log
g$. 
Note that a single vertex may correspond to multiple
satisfied/unsatisfied terminals of different groups. The following
fact implies that subgraphs in which all terminals are satisfied are
group spanners with $\beta = 2\log g$.

\begin{fact}\label{fact:sat}
  If $H$ is a spanning subgraph such that $d_H(x_{j,i}, \{x_{j,1}
  \ldots, x_{j,i-1}\}) \leq 2\log g$ for all $i > 1$, then there
  exists a Steiner tree for $X_j$ in $H$ with total size at most $(2\log g) \Steiner_G(X_j)$.
\end{fact}




Our algorithm starts with the MST $T$ and adds satisfying edges to it
in order to construct $H$. In order to bound the cost of these edges,
the algorithm maintains an arc set $E'$ defined over the vertex set
$V$. Let $G'$ denote the directed graph $(V,E')$. At the beginning of
the algorithm, $E'$ is empty. We use \emph{arcs} to refer to directed
edges in $E'$ and simply \emph{edges} for edges in $E$. Our algorithm
works in two phases. In the first phase, for each unsatisfied
terminal, the algorithm adds its satisfying edge only if we can add an
oriented copy of it to $E'$ and modify nearby arcs in $E'$ such that
the out-degree of every node is at most $2$. The main lemma is that
the number of unsatisfied terminals at the end of this phase is at
most $|V|$, and so we can simply add their satisfying edges in the
second phase. We use the following notation for the algorithm:
$\out(x)$ denotes the number of edges of $E'$ that are oriented away
from $x$; $\Gamma(x) \subseteq V$ denotes the set of terminals
reachable from $x$ via a directed path in $E'$ of length at most $\log
g$.


\begin{algorithm}
\caption{Algorithm for uniform graph spanner instances}
  \begin{algorithmic}[1]
   \label{alg:uniform}
   \STATE \emph{(Phase 1)}
   \STATE $E',A_1,A_2 \leftarrow \emptyset$
   \WHILE {there exists $x$ that is unsatisfied in $\MST
     \cup A_1$ and $z \in \Gamma(x)$ such that $\out(z) \leq 1$}
   \STATE Add $x$'s satisfying edge $(x,y)$ to $E'$ oriented from $x$ to $y$
   \STATE Add $(x,y)$ to $A_1$
   \IF {$\out(x) > 2$}
   \STATE Flip directions of arcs in $G'$ along $x-z$ path
   \ENDIF
   \ENDWHILE
   \STATE \emph{(Phase 2)}
   \STATE For every $x$ unsatisfied in $\MST \cup A_1$, add its
   satisfying edge $(x,y)$ to $A_2$
   \RETURN $A = A_1 \cup A_2$
  \end{algorithmic}
\end{algorithm}

At the end of the algorithm every vertex is
satisfied. Fact~\ref{fact:sat} then implies that $H=T\cup A_1\cup A_2$
is a group spanner with $\beta = 2 \log g$. So we only need to bound
the sizes of $A_1$ and $A_2$. Since there is a one-to-one
correspondence between edges in $A_1$ and arcs in $E'$, the following
lemma implies that $|A_1| = |E'| \leq 2|V|$. 
\begin{lemma}
  \label{lem:outdeg}
  We have $\out(x) \leq 2$ for all $x \in V$.
\end{lemma}
\begin{proof}
We prove the lemma by induction on the iterations of the algorithm.
  The base case ($E' = \emptyset$) is trivial. The interesting case is
  when $\out(x) = 2$ at the beginning of the iteration and the
  algorithm adds $(x,y)$ to $E'$ oriented from $x$ to $y$. At this
  point, we have $\out(x) = 3$, $\out(z) \leq 1$ and all other
  terminals on the $x - z$ path have out-degree at most $2$ by the
  inductive hypothesis. When the algorithm flips the arcs on the path,
  it decrements $\out(x)$ by $1$, increments $\out(z)$ by $1$ and does
  not affect the out-degrees of other terminals on the path. This proves
  the lemma.
\end{proof}

Next we bound $|A_2|$.

  


\begin{lemma}
  \label{lem:A2}
  $|A_2| \leq |V|$.
\end{lemma}

\begin{proof}
  First we prove that, even if we ignore edge directions, the length
  of the smallest cycle (i.e. girth) in $E'$ is at least $\log g$.
  Assume, towards a contradiction, that there is an undirected cycle
  of length $k \leq \log g$ in $E'$. Let $(x,y)$ be the last arc added
  in the cycle. Before the algorithm added it, there is a path from
  $x$ to $y$ of length $k-1$ in $A$ corresponding to the other arcs in
  the cycle. This contradicts the condition for adding $(x,y)$; in
  particular, $x$ is not unsatisfied.

  Let $U = \{x_{j,i} : \text{$x_{j,i}$ unsatisfied in $\MST \cup A_1$}\}$.
  For $x_{j,i} \in U$, we have $\out(z) = 2$ for all $z \in
  \Gamma(x_{j,i})$ since otherwise we would have added its satisfying
  edge in phase 1. Since the girth of $E'$ is at least $\log g$, 
  we have a full binary tree of depth $\log g$ rooted at $x_{j,i}$ in
  $E'$. This implies $|\Gamma(x_{j,i})| \geq g$. Furthermore, for any
  $x_{j,i}, x_{j,k} \in U$ with $i > k$, we have $\Gamma(x_{j,i}) \cap
  \Gamma(x_{j,k}) = \emptyset$ because otherwise 
  $d_{\MST \cup
  A_1}(x_{j,i}, x_{j,k}) \leq 2\log g$
  and $x_{j,i}$ would
  not have been unsatisfied in $\MST \cup A_1$. Therefore any terminal
  can belong to at most one $\Gamma(x_{j,i})$ per $j$, giving us
  $\sum_{x_{j,i} \in U} |\Gamma(x_{j,i})| \leq g |V|$. Hence we get
  the desired bound: $|V| \geq \sum_{x_{j,i} \in U} |\Gamma(x_{j,i})|
  / g \geq g |U| /g = |U| = |A_2|$.
  %
\end{proof}

Lemmas~\ref{lem:outdeg} and \ref{lem:A2} imply that $|A_1| + |A_2|
\leq 3 |V|$. Furthermore, the algorithm ensures that all the
terminals are satisfied in $\MST \cup A_1 \cup A_2$. Together with
Fact~\ref{fact:sat}, we get Lemma \ref{thm:uniform}.
